%% file: main.tex
\documentclass[preliminary]{eptcs}

\usepackage{amsmath,amssymb,amsthm,wasysym}
\usepackage{latexsym}
\usepackage{xspace}
\usepackage{pgf}
\usepackage{tikz}
\usetikzlibrary{automata,decorations.pathreplacing,fit,calc}
\usepackage{multirow}

\input{macros}

\title{Buffered Simulation Games for B\"uchi Automata}
\def\titlerunning{Buffered Simulation Games for B\"uchi Automata}

\author{Milka Hutagalung, Martin Lange and Etienne Lozes  
\institute{School of Electr.\ Eng.\ and Computer Science, University of Kassel, Germany\thanks
{The European Research Council has provided financial support under the European Community's Seventh Framework Programme (FP7/2007-2013) / ERC grant agreement no 259267.}}}
\def\authorrunning{M.~Hutagalung, M.~Lange, E.~Lozes}

\allowdisplaybreaks

\begin{document}

\maketitle

\begin{abstract}
Simulation relations are an important tool in automata theory because they provide 
efficiently computable approximations to language inclusion. In recent years, extensions
of ordinary simulations have been studied, for instance multi-pebble and multi-letter
simulations which yield better approximations and are still polynomial-time computable.

In this paper we study the limitations of approximating language inclusion in this way: we 
introduce a natural extension of multi-letter simulations called buffered simulations. They
are based on a simulation game in which the two players share a FIFO buffer of unbounded size. We consider two variants of these buffered games called continuous and look-ahead 
simulation which differ in how elements can be removed from the FIFO buffer. We show that 
look-ahead simulation, the simpler one, is already PSPACE-hard, i.e.\ computationally as hard 
as language inclusion itself. Continuous simulation is even EXPTIME-hard. We also provide 
matching upper bounds for solving these games with infinite state spaces.
\end{abstract}

\addtocounter{footnote}{-1}

\section{Introduction}
\input{introduction.tex}
\section{Extended Simulation Relations}
\label{sec:extsim}
\subsection{Background}
\input{background.tex}
\subsection{Continuous Simulation\label{sec:continuous}}
\input{continuous-def.tex}
\subsection{Look-Ahead Simulations}
\input{look-ahead.tex}
\section{Lower Bounds: The Complexity of Buffered Simulations}
\label{sec:complexity}
\input{exptime.tex}
\section{Upper Bounds: Quotient Games}
\label{sec:upper}
\input{quotient.tex}
\section{Properties of Buffered Simulations}
\label{sec:properties}
\input{properties.tex}
\paragraph*{Buffered Simulations in Automata Minimisation.}
\input{conclusion.tex}

\bibliographystyle{eptcs}
\bibliography{refs-with-doi}


\end{document}

%% file: macros.tex
\newcommand{\ar}[1]{\stackrel{#1}{\longrightarrow}}
\newcommand{\aro}{
\setbox0\hbox{$\longrightarrow$}%
    \rlap{\hbox to \wd0{\hss$\bell$\hss}}\box0}
\newcommand{\aar}[1]{\mathop{\stackrel{#1}{\aro}}}
\newcommand{\ARuns}{\mathsf{ARuns}}

\newcommand{\look}{\mathsf{la}}
\newcommand{\cont}{\mathsf{co}}

\newcommand{\fair}{\mathsf{f}}

\renewcommand{\epsilon}{\varepsilon}

\newcommand{\ba}{\mathcal A}  
\newcommand{\babis}{\mathcal B}

\def\newarrow#1{\mathop{{\hbox{\setbox0=\hbox{$\scriptstyle{#1\quad}$}{$%
\mathrel{\mathop{\setbox1=\hbox to
\wd0{\rightarrowfill}\ht1=3pt\dp1=-2pt\box1}\limits^{#1}}%
$}}}}}
\newcommand{\nat}{\mathbb N}  

\newcommand{\Nat}{\nat}

\newcommand{\Runs}{\mathop{\mathsf{Runs}}}

\newtheorem{theorem}{Theorem}[section]
\newtheorem{lemma}[theorem]{Lemma}
\newtheorem{proposition}[theorem]{Proposition}
\newtheorem{corollary}[theorem]{Corollary}
\theoremstyle{definition}
\newtheorem{definition}[theorem]{Definition}
\newtheorem{example}[theorem]{Example}
\newtheorem{remark}[theorem]{Remark}


%% file: introduction.tex
Nondeterministic B\"uchi automata (NBA) are an important formalism for the
specification and verification of reactive systems. 
While they have originally been
introduced as an auxiliary device in the quest for a decision procedure for 
Monadic Second-Order Logic \cite{buc62} they are by now commonly used
in such applications as
LTL software model-checking~\cite{EtessamiH00,spin-book},
or size-change 
termination analysis for recursive 
programs~\cite{size-change-termination,fogarty-vardi-2012}. Typical decision procedures 
from these domains then reduce to automata-theoretic decision problems like emptiness
or inclusion for instance \cite{Vardi96}.

While emptiness for B\"uchi automata is NLOGSPACE-complete, deciding inclusion between two
nondeterministic finite automata is already more difficult, namely PSPACE-complete \cite{STOC73*1}.
This is also the complexity of inclusion for NBA. Thus, it is -- given current knowledge -- 
exponential in the size of the involved NBA regardless of whether it is solved using 
explicit complementation \cite{FOCS88*319,thomas-salomaa,KupfermanV01} or other means 
\cite{AbdullaCCHHMV10,conf/tacas/FogartyV10}.
One major issue of automata manipulation is therefore to keep the number
of states as small as possible. 

Since the early works of Dill~et~al~\cite{DillHW91}, simulations
have been intensively used in automata-based verification. 
Unlike the PSPACE-hard problems like inclusion, simulation between two NBA is cheap to compute.
Simulations are interesting with respect
to several aspects.
On the one hand, they offer a sound, but incomplete, approximation of
language inclusion that may be sufficient in many practical cases.
On the other hand, simulations can be used for 
quotienting 
automata~\cite{Bustan:2003:SM:635499.635502,fairsimmini,EtessamiWS01}, 
for pruning transitions~\cite{AbdullaCCHHMV10,AbdullaCCHHMV11},
or for improving existing decision procedures on NBA like the Ramsey-based \cite{fogarty-vardi-2012} or the 
antichain algorithm for inclusion, resp.\ universality checking \cite{antichain-journal}. 

There is a simple game-theoretic characterisation of simulation between two NBA: two players
called Spoiler and Duplicator move two pebbles on the transition graph of the NBA, each of
them controls one pebble. In order to decide whether or not an NBA $\mathcal{A}$ is
simulated by an NBA $\mathcal{B}$, Spoiler starts with his pebble on the initial state of
$\mathcal{A}$ and moves it along a transition labeled with some alphabet symbol $a$. Duplicator 
starts with her pebble on the initial state of $\mathcal{B}$ and responds with a move along 
a transition labeled with the same letter. This proceeds ad infinitum. There are different
kinds of simulation depending on the winning conditions in these games. For instance, fair
simulation models the B\"uchi acceptance condition and requires Duplicator to have visited
infinitely often accepting states if Spoiler has done so. While it is close to the actual
condition on inclusion between these two automata, quotienting automata with respect to fair 
simulation does not preserve the automaton's language.

It is therefore that different winning conditions like delayed simulation have been invented
which require Duplicator to eventually visit an accepting states whenever Spoiler has visited
one \cite{EtessamiWS01}. They, however, do not necessarily provide better approximations to
language inclusion. Extensions of the plain simulation relation have been considered since,
in particular multi-pebble \cite{pebble} and multi-letter simulations \cite{MayrC13,HutagalungLL13}.  
Both try to alleviate the gap between simulation and language inclusion which shows up in the
game-theoretic characterisation as Spoiler being too strong: language inclusion would correspond
to a game in which player chooses an entire run in $\mathcal{A}$ and then Duplicator produces
one in $\mathcal{B}$ on the same word. In the simulation game, Spoiler reveals his run step-wise
and can therefore dupe Duplicator into positions from which she cannot win anymore even though
language inclusion holds.

The two extensions -- multi-pebble and multi-letter simulation -- use different approaches to 
approximate language inclusion better: multi-pebble simulation add a certain degree of imperfectness
to these games by allowing Duplicator to be in several positions at the same time. Multi-letter
simulation forces Spoiler to reveal more of his runs and therefore allows Duplicator to delay
her choices for a few rounds and therefore benefit from additional information she gained about
Spoiler's moves. The complexity of computing these extended simulations has been studied before:
both are polynomial for a fixed number of pebbles, respectively a fixed look-ahead in the multi-letter
games. However, nothing is known about the complexity of these simulations
if the number of pebbles/letters is not fixed.

\paragraph*{Contribution.}
This paper studies a natural extension of multi-letter games to unbounded look-aheads. We introduce a 
new family of simulation relations for B\"uchi automata, called buffered simulations. In a buffered 
simulation, Spoiler and Duplicator move two pebbles along automata transitions, but unlike in standard
simulations, Spoiler and Duplicator's moves do not always alternate. 
Indeed, Duplicator can ``skip her turn'' and wait to see Spoiler's next
moves before responding. 
Spoiler and Duplicator share a first-in first-out buffer:
every time Spoiler moves along an $a$-labelled transition, he adds an
$a$ into the buffer, whereas every time 
Duplicator makes a step along a $b$-labelled transition, she removes
a $b$ from the buffer. 
Since Duplicator has more chances to defeat Spoiler than in standard
simulations, buffered simulations better approximate language inclusion. They
also improve multi-letter simulations, and it is thus a natural
question to ask if they are polynomial time decidable and could be used
in practice.

We study two notions of buffered simulation games, called
continuous and look-ahead simulation games, respectively. Their rules
only differ in the way that Duplicator must use the buffer: in look-ahead
simulations, Duplicator is forced to flush the buffer, so that she 
``catches up''
with Spoiler every time she decides to make a move. Thus, the buffer is flushed
completely with each of Duplicator's moves. In the continuous case, Duplicator
can choose to only consume a part of the buffer with every move, and it need
not ever be flushed.


We show that these unbounded buffer simulation games -- whilst naturally extending
the ``easy'' multi-letter simulations -- provide in a sense a limit to the
efficient approximability of language inclusion: we show that look-ahead simulations
are already PSPACE-hard, i.e.\ as difficult as language inclusion itself, while
continuous simulations are even worse: they are EXPTIME-hard, i.e.\ presumably
even more difficult than language inclusion.

We also provide matching upper bounds in order to show that these lower bounds
are tight, i.e.\ these simulations problems are not worse than that. In particular,
look-ahead simulation is therefore as difficult as language inclusion, and continuous
simulation is ``only'' slightly more difficult. Decidability of these simulations is 
not obvious. In the finitary cases, it is provided by a rather straight-forward 
reduction to parity games but games with unbounded buffers would yield parity games 
of infinite size. Moreover, questions about systems with unbounded FIFO buffers are 
often undecidable; for instance, linear-time properties of a system of two machines 
and one buffer are known to be undecidable~\cite{CeceF05}. Decidability of these
simulation relations may therefore be seen as surprising, and it is also not
inconceivable that the decidability results for these unbounded FIFO buffer 
simulations may lead to developments in other areas, for instance reachability in
infinite-state systems etc.


\paragraph*{Outline.}
Section~\ref{sec:extsim} first recalls B\"uchi automata and ordinary simulation
relations. It then introduces continuous simulation as a simulation game extended
with an unbounded buffer. Look-ahead simulation is obtained by restricting the use
of the buffer in a natural way. Section~\ref{sec:complexity} contains the most
important results in these relations: it shows that look-ahead simulation is already
as hard as language inclusion whereas continuous simulation is even harder. 
Section~\ref{sec:upper} shows that these bounds are tight by introducing a suitable
abstraction called quotient game which yields corresponding upper bounds. Finally,
Section~\ref{sec:properties} collects further interesting results on these 
simulation relations like topological characterisations for instance and concludes
with comments on their use in automata minimisation.


%% file: background.tex
\paragraph*{Nondeterministic B\"uchi Automata.}
A non-deterministic B\"uchi automaton (NBA) 
is a tuple $\ba=(Q,\Sigma,\delta,q_0,F)$
where $Q$ is a finite set of states with $q_0$ being a designated 
starting state, $\delta\subseteq
Q\times \Sigma\times Q$ is a transition relation, and
$F\subseteq Q$ is a set of accepting states. 
A state $q \in Q$ is called a \emph{dead end} when there is no 
$a \in \Sigma$ and $q' \in Q$ such that $(q,a,q') \in \delta$.
If $w=a_1\dots a_{n}$, a 
sequence $q_0a_1q_1\dots q_n$ is called a $w$-path from $q_0$ to
$q_n$ if $(q_i,a_{i+1},q_{i+1})\in\delta$ for all $i\in\{0,\dots,n-1\}$. It is an \emph{accepting}
$w$-path if there is some $i\in\{1,\dots,n\}$ such that $q_i\in F$.  We write 
$q_0\ar{w} q_n$ to state that there is a $w$-path from $q_0$ to $q_n$, and 
$q_0\aar{w}q_n$ to state that there is an accepting one.

A \emph{run} of $\ba$ on a word $w=a_1a_2\dots\in \Sigma^{\omega}$ is an infinite
sequence $\rho=q_0 q_1 \dots$ such that $(q_i,a_{i+1},q_{i+1})\in \delta$ 
for all $i\geq 0$. The run is \emph{accepting} if there is some $q\in F$ 
such that
$q=q_i$ for infinitely many $i$.
The language of $\ba$ is the set $L(\ba)$ of infinite words for which
there exists an accepting run. 

\subsubsection*{Fair Simulation.}
Fair simulation~\cite{HenzingerKR02} is an extension of standard simulation
to B\"uchi automata. The easiest way of defining fair simulation 
is by means of a game
between two players called \emph{Spoiler} and \emph{Duplicator}. 
Let us fix two NBA $\ba =
(Q,\Sigma,\delta,q_I,F)$ and $\babis=(Q',\Sigma,\delta',q_I',F')$.
Spoiler and Duplicator are each given a pebble that is initially placed
on $q_{0}:=q_I$ for Spoiler and $q'_{0}:=q_I'$ for Duplicator. Then, on each round 
$i\geq1$,
\begin{enumerate}
 \item Spoiler chooses a letter $a_i \in \Sigma$ and a transition $(q_{i-1},a_i,q_i) \in \delta$, and moves his 
 pebble to $q_i$;
 \item Duplicator responds by choosing a transition $(q'_{i-1},a_i,q'_i) \in \delta'$ and moves his pebble to $q'_i$.
\end{enumerate}
Either the play terminates because one player reaches a dead end, 
and then the opponent wins the play. Or the game produces two
infinite runs $\rho = q_0a_1q_1,\ldots$ and 
$\rho' = q'_0a_1q'_1\ldots$, in which case Duplicator is declared the winner of the play 
if $\rho$ is not accepting or $\rho'$ is accepting. Otherwise Spoiler wins this play.

We say that $\ba$ is \emph{fairly simulated} by $\babis$, written $\ba\sqsubseteq^{\fair}\babis$,
if Duplicator has a winning strategy for this game. 
Clearly, $\ba\sqsubseteq^{\fair}\babis$
implies $L(\ba)\subseteq L(\babis)$, but the converse does not hold 
in general. 

\begin{remark}
Notice that standard simulation, as defined for labelled transition systems,
is a special case of fair
simulation. Indeed, for a given labelled transition system
$(Q,\Sigma,\delta)$, and a given state $q$, we can define the NBA $\ba(q)$ with
$q_I:=q$ as the initial state, and $F:=Q$ as the set of accepting states.
Then $q'$ simulates $q$ in the standard sense (without taking care
of fairness) if and only if $\ba(q)\sqsubseteq^{\fair}\ba(q')$. We write
$q\sqsubseteq q'$ when $q'$ simulates $q$ in the standard sense.
\end{remark}



%% file: continuous-def.tex
Continuous simulations are defined by games in which 
Duplicator is allowed to see in advance some
finite but \emph{unbounded} number of Spoiler's moves. This naturally extends recent work
on extensions of fair simulation called \emph{multi-letter} or \emph{look-ahead simulations} in
which Duplicator is allowed to see a number of Spoiler's moves that is \emph{bounded} by a
constant \cite{HutagalungLL13,MayrC13}.

Let $\ba=(Q,\Sigma,\delta,q_I,F)$ and $\babis=(Q',\Sigma,\delta',q_I',F')$ be two NBA.  In the 
continuous fair simulation game, Spoiler and Duplicator now share a FIFO buffer $\beta$
and move two pebbles through the automata's state spaces.  The positions of the pebbles form 
a word $w$ and two runs $\rho$ and $\rho'$, obtained by successively extending sequences
$\rho_i$ and $\rho'_i$ in each round $i$ with zero or more states. At the beginning we have 
$\rho_0 := q_I$ and $\rho_0' := q'_I$, i.e.\ Spoiler's pebble is on $q_I$ and Duplicator's
pebble is on $q_I'$. Initially, both word and buffer are empty, i.e.\ we have 
$w_0 := \epsilon$ and $\beta_0 := \epsilon$.

For the $m$-th round, with $m \ge 1$ suppose that $w_{m-1} = a_1,\ldots,a_{m-1}$, 
$\rho_{m-1} = q_0,\ldots,q_{m-1}$, $\rho'_{m-1} = q_0',\ldots,q_{m'}$ and 
$\beta_{m-1}$ have
been created already. Duplicator's run in $\babis$
is shorter than Spoiler's run, i.e. $m' \le m$.
Furthermore, the buffer $\beta$ contains the suffix $a_{m'+1},\dots,a_{m}$
of $w_{m-1}$ that Duplicator has not mimicked yet. The $m$-th round then 
proceeds as follows.

\begin{enumerate}
 \item Spoiler chooses a letter $a_{m} \in \Sigma$ and a transition 
 $(q_{m-1},a_{m},q_{m}) \in \delta$ and moves the 
 pebble to $q_{m}$, i.e.\ we get $w_{m} := w_{m-1} a_{m}$ and $\rho_m: = \rho_{m-1} q_{m+1}$. The 
letter $a_{m}$ is added to the buffer, i.e.\ $\beta' := \beta, a_{m}$. 
 \item Suppose we now have $\beta' = b_1, \ldots, b_k$. Duplicator picks some $r$ with 
$0 \le r \le k$ as well as states $q'_{m'},\ldots,q'_{m'+r-1}$ such that
$(q_{m'+i-1},b_i,q_{m'+i}) \in \delta'$ for all $i=1,\ldots,r$. Then we get 
$\rho'_{m} := \rho'_{m-1}, q_{m'+1},$ $\ldots, q_{m'+r}$. The letters get flushed from the buffer, i.e.\ 
$\beta_i := b_{r+1},\ldots,b_k$. 

Note that we have $\rho'_m = \rho'_{m-1}$ if Duplicator chooses $r=0$. In this case we also say
that she \emph{skips her turn}.
\end{enumerate}
A play of this game
defines a finite or infinite run $\rho = q_0,q_1,\ldots$ for Spoiler (finite
if Spoiler reaches a dead end), and a
finite or infinite run $\rho' = q'_0,q'_1,\ldots$ for Duplicator 
(finite if Duplicator eventually always skips her turn) on the finite or infinite word 
$w = a_1 a_2 \ldots$. 

Duplicator is declared the winner of the play if
\begin{itemize}
\item $\rho$ is finite, or
\item $\rho$ is infinite (and $w$ is necessarily infinite as well) and 
         \begin{itemize}
           \item $\rho$ is not an accepting run on $w$, or 
           \item $\rho'$ is infinite and an accepting run on $w$.
         \end{itemize}
\end{itemize}
In all other cases, Spoiler wins the play.

We say that $\babis$ \emph{continuously fairly simulates} $\ba$, 
written $\ba\sqsubseteq_{\cont}^{\fair}\babis$, if Duplicator has a winning strategy for the 
continuous fair simulation game on $\ba$ and $\babis$. We also
consider the (unfair) continuous simulation $\sqsubseteq_{\cont}$
for pairs of LTS states by considering an LTS with a distinguished state
as an NBA where all states are accepting.

\begin{example}
\label{ex:branching}
Consider the following two NBA $\ba$ (left) and $\babis$ (right) over
the alphabet $\Sigma=\{a,b,c\}$.
\begin{center}
\begin{tikzpicture}[initial text={}, node distance=12mm, every state/.style={minimum size=5mm}, thick]
  \node[state,initial]   (qa0)                      {};
  \node[state]           (qa1) [right of=qa0]       {};
  \node                  (aux1) [right of=qa1]      {};
  \node[state,accepting] (qa2) [above of=aux1, node distance=4mm] {};
  \node[state,accepting] (qa3) [below of=aux1, node distance=4mm] {};

  \path[->] (qa0) edge              node [above]      {$a$}     (qa1)
            (qa1) edge              node [above left] {$b$}     (qa2)
                  edge              node [below left] {$c$}     (qa3)
            (qa2) edge [loop right] node              {$\Sigma$} ()
            (qa3) edge [loop right] node              {$\Sigma$} ()
            (qa1) edge [loop above] node              {$a$}      ();      

  \node[state,initial]   (qb0) [right of=qa1, node distance=5cm] {};
  \node                  (aux2) [right of=qb0]                   {};
  \node[state]           (qb1) [above of=aux2, node distance=4mm] {};
  \node[state,accepting] (qb2) [right of=qb1]                    {};
  \node[state]           (qb3) [below of=aux2, node distance=4mm] {};
  \node[state,accepting] (qb4) [right of=qb3]                    {};

  \path[->] (qb0) edge              node [above left] {$a$}     (qb1)
                  edge              node [below left] {$a$}     (qb3)
            (qb1) edge              node [above]      {$b$}     (qb2)
            (qb2) edge [loop right] node              {$\Sigma$} ()
            (qb3) edge              node [below]      {$c$}     (qb4)
            (qb4) edge [loop right] node              {$\Sigma$} ()
            (qb1) edge [loop above] node              {$a$} ()
            (qb3) edge [loop below] node              {$a$} ();

\end{tikzpicture}
\end{center}
Clearly, we have $L(\ba) \subseteq L(\babis)$.

Duplicator has a winning strategy for the continuous fair 
simulation game on this pair of automata: she skips her turns until Spoiler 
follows either the $b$- or the $c$-transition. However, if we ignore the accepting
states and consider these automata as a transition system,
then Spoiler has a winning strategy for the continuous simulation:
he iterates the $a$-loop, and then either Duplicator waits forever
and loses the play, or she makes a move and it is then easy for
Spoiler to defeat her.

This example also shows that continuous fair simulation strictly extends multi-letter fair
simulation which can be seen as the restriction of the former to a bounded buffer. I.e.\ in these
games, Duplicator can only benefit from a fixed look-ahead of at most $k$ letters for some $k$.
It is not hard to see that Spoiler wins the game with a bounded buffer of length $k$ for any $k$
on these two automata: he simply takes $k$ turns on the $a$-loop in $\ba$ which forces
Duplicator to choose a transition out of the initial state in $\babis$. After doing so, Spoiler
can choose the $b$- or $c$-transition that is not present for Duplicator anymore and make her
get stuck.
\end{example}


%% file: look-ahead.tex
We now consider a variant of the continuous simulation games 
called \emph{look-ahead} simulation games (the terminology follows \cite{MayrC13}). 
Look-ahead simulation games proceed
exactly like the continuous ones, except that now Duplicator
has only two possibilities: either she skips her turn, or she flushes 
the entire buffer. Formally, the definition of the game
only differs from the one of Section~\ref{sec:continuous} in that the 
number $r$ of letters removed by Duplicator in a round is either $0$ or
the size $|\beta|$ of the current buffer $\beta$, whereas continuous simulation
allowed any $r\in\{0,\dots,|\beta|\}$.

We write $\ba\sqsubseteq_{\look}^{\fair}\babis$ 
if Duplicator has a winning strategy for the look-ahead fair simulation 
on the two automata $\ba,\babis$. 
Similarly, we define the look-ahead fair simulation for LTS, 
$\sqsubseteq_{\look}$.

\begin{example}
Consider again $\ba$ and $\babis$ as in Example~\ref{ex:branching}.
It holds that $\ba\sqsubseteq_{\look}^{\fair}\babis$,
because Duplicator can flush the buffer once she has seen the first $b$ or $c$.
\end{example}

Clearly, look-ahead simulation implies continuous simulation but the converse does not hold.

\begin{example}
Consider the following two NBA $\ba$ (left) and $\babis$ (right) 
over the alphabet
$\Sigma = \{a,b,c\}$.
\begin{center}
\begin{tikzpicture}[initial text={}, node distance=12mm, every state/.style={minimum size=5mm}, thick]
  \node[state,initial]   (qa0)                      {};
  \node[state,accepting] (qa1) [right of=qa0]       {};

  \path[->] (qa0) edge              node [above]      {$a$}     (qa1)
            (qa1) edge [loop right] node [right]      {$b,c$}     ();

  \node[state,initial]   (qb0) [right of=qa1, node distance=5cm] {};
  \node                  (aux2) [right of=qb0]                   {};
  \node[state,accepting] (qb1) [above of=aux2, node distance=6mm] {};
  \node[state,accepting] (qb3) [below of=aux2, node distance=6mm] {};

  \path[->] (qb0)                    edge              node [above left] {$a$}     (qb1)
                                     edge              node [below left] {$a$}     (qb3)
            (qb1)                    edge [loop right] node [right]      {$b$}     ()
            ([xshift=-1mm]qb1.south) edge              node [left]       {$b$}     ([xshift=-1mm]qb3.north)
            (qb3)                    edge [loop right] node [right]      {$c$}     ()
            ([xshift=1mm]qb3.north)  edge              node [right]      {$c$}     ([xshift=1mm]qb1.south);
\end{tikzpicture}
\end{center}
Duplicator wins the continuous fair simulation: a 
winning strategy for Duplicator 
is to skip her first turn, and then to remove one letter at a time during the
rest of the play. Thus, after each round, the buffer always contains exactly one element.

On the other hand, Spoiler wins the look-ahead simulation,
because the first time Duplicator flushes the buffer, she has committed
to a choice between the two right states and thus makes a prediction
about the next letter that Spoiler will play.
\end{example}

\begin{remark}
Multi-pebble simulations~\cite{pebble} are another notion of 
simulation in which duplicator is given more than just one pebble,
which she can move, duplicate, and drop during the game. If the number
of such pebbles is not bounded, 
multi-pebble simulations better approximate language inclusion
than continuous and look-ahead simulation; in particular,  
the look-ahead simulation game corresponds to the multi-pebble simulation
game in which duplicator is required to drop all but one pebble infinitely
often. 
\end{remark}


%% file: exptime.tex
\label{sec:exptime}
The difficulty of deciding continuous and look-ahead simulation is shown by reduction from suitable tiling
problems.

\begin{definition}
A tiling system is a tuple $\mathcal{T} = (T,\mathbin{H},\mathbin{V},t_I,t_F)$,
where $T$ is a set of tiles, 
$\mathbin{H},\mathbin{V} \subseteq T \times T$ are the horizontal and vertical 
compatibility relations, $t_I,t_F \in T$ are the initial and final tiles.

Let $n,m$ be two natural numbers. A tiling with $n$ columns and $m$ rows
according to $\mathcal T$
is a function $t:\{1,\dots,n\}\times\{1,\dots,m\}\to T$; the tiling is valid 
if (1) $t_{1,1}=t_I$ and $t_{n,m} = t_F$, 
(2) for all $i=1,\dots,n-1$ and all $j=1,\dots,m$ we have 
$(t_{i,j},t_{i+1,j})\in\mathbin{H}$, 
(3) for all $i=1,\dots,n$, for all $j=1,\dots, m-1$ we have 
$(t_{i,j},t_{i,j+1})\in \mathbin{V}$.
\end{definition}
 
The problem of deciding whether there exists a valid tiling with $n$ columns 
and $2^n$ rows, for a given $n$ in unary and a tiling system $\mathcal T$, is 
known to be PSPACE-hard \cite{Boas97theconvenience}.\footnote{The requirement on the final tile
for instance is not needed for PSPACE-hardness but this variant of the tiling problem is most
convenient for the reductions presented here.} Clearly, the problem to decide whether there is
no such tiling is equally PSPACE-hard. We reduce the complement of the tiling problem to 
look-ahead buffered simulation.

\begin{figure}[t]
\begin{tikzpicture}[every state/.style={minimum size=5mm}, thick]
                            
\begin{scope}[node distance=3cm]
  \node[state,initial,initial text={$\ba$}]          (0)  {};
  \node[state]           (01) [right of=0,node distance=1.5cm]     {};
  \node[state]           (1) [right of=01,node distance=1.5cm]     {};
  \node[state] (2) [right of=1] {};
  \node[state] (3) [right of=2] {};
  \node[state] (34) [right of=3] {};
  \node[state] (4) [right of=34,node distance=1.5cm] {};

  \path[->] 
  (0) edge node [above] {\$}     (01)
  (01) edge node [above] {$t_I,0$} (1)
  (1) edge node [above] {$(T\times \{0\})^{n-1}$}     (2)
  (2) edge [bend left] node [above] {\$} (3)
  (3) edge [bend left] node [below] {$(T\times \{0,1\})^{n}~\setminus~(T\times \{1\})^{n}$}  (2)
  (3) edge node [above] {$(T\times \{1\})^{n-1}$}     (34)
  (34) edge node [above] {$t_F,1$}     (4)
  (4) edge [loop right] node {$\#$} ();
\end{scope}

\begin{scope}[node distance=3cm]
  \node[state,initial,initial text={$\babis$}]  (3) [below of=0,node distance=2cm] {}; 
  \node[state] (4) [right of=3] {\tiny $q_0$};
  \node[state] (qt0) [below of=4,node distance=4.5cm] {\tiny $q_1$};
  \node[state] (qtt0) [below of=qt0,node distance=1.8cm] {\tiny $q_r$};
  \node[state] (qt1) [right of=qt0,node distance=1cm] {};
  \node[state] (qtt1) [right of=qtt0,node distance=1cm] {};
  \node[state] (qt2) [right of=qt1,node distance=6cm] {};
  \node[state] (qtt2) [right of=qtt1,node distance=6cm] {};
  \node[state] (5) [right of=4,node distance=2cm] {};
  \node[state] (51) [right of=5,node distance=2cm] {};
  \node[state] (52) [right of=51,node distance=2cm] {};
  \node[state] (6) [below of=5,node distance=1cm] {};
  \node[state] (61) [right of=6,node distance=2cm] {};
  \node[state] (62) [right of=61,node distance=2cm] {};
  \node[state] (7) [below of=6,node distance=2cm] {};
  \node[state] (8) [below of=61,node distance=1cm] {};
  \node[state] (9) [below of=8,node distance=1cm] {};
  \node[state] (81) [right of=8,node distance=2cm] {};
  \node[state] (82) [right of=81,node distance=2cm] {};
  \node[state] (91) [right of=9,node distance=2cm] {};
  \node[state] (92) [right of=91,node distance=2cm] {};

  \path[->] 
  (3) edge [loop above] node [above] {$\Sigma\setminus\{\#\}$} ()
  (3) edge node [above] {\$} (4)
  (3) edge [bend right] node [right] {$t_1,\_$} node [below left] {$\vdots$} (qt0)
  (3) edge [bend right] node [left] {$t_r,\_$} (qtt0)
  (qt0) edge node [above] {$\Sigma^n$}  (qt1)
  (qtt0) edge node [above] {$\Sigma^n$}  (qtt1)
  (qt1) edge node [above] {$\{(t,i)\mid(t_1,t)\not\in V,i\in\{0,1\}\}$}  (qt2)
  (qtt1) edge node [above] {$\{(t,i)\mid(t_r,t)\not\in V,i\in\{0,1\}\}$}  (qtt2)
  (qt2) edge [loop right] node [right] {$\Sigma$}     ()
  (qtt2) edge [loop right] node [right] {$\Sigma$}     ()
  (qt0) [draw] -- ($(qt0)+(0,-.8cm)$) -- node [above] {$\{(t,i)\mid(t_1,t)\not\in H,i\in\{0,1\}\}$} ($(qt2)+(0,-.8cm)$) 
  ($(qt2)+(0,-.8cm)$) edge (qt2)
  (qtt0) [draw] -- ($(qtt0)+(0,-.8cm)$) -- node [above] {$\{(t,i)\mid(t_r,t)\not\in H,i\in\{0,1\}\}$} ($(qtt2)+(0,-.8cm)$) [draw,->] -- (qtt2)
  (4) edge [loop above] node [above] {$\_,1$} ()
  (4) edge node [above] {$\_,0$} (5)
  (5) edge node [above] {$\Sigma^n$} (51)
  (51) edge node [above] {$\_,0$} (52)
  (52) edge [loop right] node [right] {$\Sigma$} (52)
  (4) edge [bend right] node [above right] {$\_,1$} (6)
  (6) edge node [above] {$\Sigma^n$} (61)
  (61) edge node [above] {$\_,1$} (62)
  (62) edge [loop right] node [right] {$\Sigma$} (62)
  (4) edge [bend right] node [left] {$\_,0$} (7)
  (7) edge [loop above] node [above] {$\Sigma\setminus\{\$\}$} ()
  (7) edge node [above] {$\_,0$} (8)
  (7) edge node [above] {$\_,1$} (9)
  (8) edge node [above] {$\Sigma^n$} (81)
  (81) edge node [above] {$\_,1$} (82)
  (82) edge [loop right] node [right] {$\Sigma$} ()
  (9) edge node [above] {$\Sigma^n$} (91)
  (91) edge node [above] {$\_,0$} (92)
  (92) edge [loop right] node [right] {$\Sigma$} ()
  
  ;
\end{scope}
\end{tikzpicture}
\caption{Automata $\ba$ and $\babis$ 
used in the proof of Theorem~\ref{pspace}.}
\label{fig:autpspace}
\end{figure}

\begin{theorem}\label{pspace}
Deciding $\sqsubseteq_{\look}$ (resp. $\sqsubseteq^{\fair}_{\look}$) 
is PSPACE-hard.
\end{theorem}

\begin{proof}
Given a tiling system $\mathcal{T} = (T,H,V,t_I,t_F)$
and an $n \in \Nat$, we consider the alphabet $\Sigma:= (T\times \{0,1\})\cup\{\$,\#\}$.
We  define the two automata $\ba$, $\babis$ as depicted on
Figure~\ref{fig:autpspace}, where all states are accepting.
The sizes of $\ba$, $\babis$
are polynomial in $|T|+n$.
Let us consider first the automaton $\ba$.
A word accepted by $\ba$ is composed of blocks of $n$ tiles separated
by the $\$$ symbol, such that each block is tagged
with the binary representation of a number in $\{0,\dots,2^n-1\}$.
We take as a convention that the first bit is the least significant one.
Either the word contains finitely many blocks, in which case, the word ends with
the symbol $\#$ repeated infinitely often, or it contains infinitely many blocks.
Moreover, the first block is tagged with 0, and the last one, if it exists, is
tagged with $2^n-1$ and it is the only one that may be tagged with $2^n-1$.

Consider now the automaton $\babis$. From state $q_0$, the automaton accepts a 
word if the two first blocks are not tagged with consecutive numbers.
From the state $q_i$, the automaton accepts a word if either it starts with a
tile that is not horizontally compatible with $t_i$, or if after $n$ symbols
it contains a tile that is not vertically compatible with $t_i$.

The claim is that $\ba\sqsubseteq_{\look}\babis$ (resp. $\ba\sqsubseteq_{\look}^{\fair}\babis$) if and only if there is no valid $n\times 2^n$ tiling.
Assume first that a valid tiling exists. Then Spoiler wins if he plays
the word that contains in the $i$-th block the $i$-th row of the tiling tagged
with the binary representation of $i$. Note that Duplicator cannot
loop forever in the initial state because she cannot read the $\#$ symbol.
Conversely, assume there is
no valid tiling. Then Duplicator wins if she waits
until she has seen at most $2^n+1$ blocks: either two blocks are not  
tagged with consecutive numbers, or Spoiler played exactly $2^n$ blocks
but these do not code a valid tiling. In the former, Duplicator
then accepts by moving to $q_0$ at the beginning of the first ill-tagged block,
and in the later, she wins by moving to $q_i$ after having read a
tile $t_i$ whose horizontal or vertical successor does not match.
\end{proof}

\begin{figure}[t]
\center
\begin{tikzpicture}[initial text={}, node distance=12mm, every state/.style={minimum size=5mm}, thick]
    \node[state, initial below] (t1) {\tiny $q_{t_1}$};
    \node[state] (t2) [below of =t1, node distance = 2cm] {\tiny $q_{t_2}$};
    \node[state] (t3) [below of = t2,node distance = 2cm] {\tiny $q_{t_3}$};
    
    \node[state] (0t1) [left of = t1] {};
    \node[state] (0t2) [left of =t2] {};
    \node[state] (0t3) [left of =t3] {};
    
    \node[state] (ft1) [left of = 0t1] {};
    \node[state] (ft2) [left of =0t2] {};
    \node[state] (ft3) [left of =0t3] {};
    \node (label) [above of= ft1, node distance = 1.8cm] {(P1)};
    \node (label) [above of= t1, node distance = 1.8cm] {(P4), (P5)};    

    \path[->]
       (t1) edge [bend left =20 ]             node [right] {$0$}     (t2)
      (t2) edge [bend left =20 ]             node [left] {$0$}     (t1)
      (t2) edge [bend left =20 ]             node [right] {$0$}     (t3)
      (t1) edge [loop above]             node [above] {$T \backslash t_F,1$}     ()
      (t2) edge [loop below]             node [left] {$T \backslash t_F,1$}     ()
      (t3) edge [loop below]             node [below] {$T \backslash t_F,1$}     ()
      (t1) edge              node [above] {$0$}     (0t1)
      (t2) edge              node [above] {$0$}     (0t2)
      (t3) edge              node [above] {$0$}     (0t3)
      (0t1) edge              node [above] {$T \backslash t_1$}     (ft1)
      (0t2) edge              node [above] {$T \backslash t_2$}     (ft2)
      (0t3) edge              node [above] {$T \backslash t_3$}     (ft3)
      (ft1) edge [loop above]             node [above] {$\Sigma$}     ()
      (ft2) edge [loop above]             node [above] {$\Sigma$}     ()
      (ft3) edge [loop below]             node [below] {$\Sigma$}     ();    
    
    \node[state]  (init_a) [right of=t1, node distance =3cm] {}; 
    \node[state]          (21) [right of=init_a] {};
    \node[state]          (11) [above of=21] {};
    \node[state]          (12) [right of=11] {};
    \node[state]          (13) [right of=12] {};
    \node[state]          (14) [right of=13] {};
    \node[state]          (22) [right of=21] {};
    \node[state]          (23) [right of=22] {};
    \node[state]          (24) [right of=23] {};
    \node[state]          (31) [below of=21] {};
    \node[state]          (32) [right of=31] {};
    \node[state]          (33) [right of=32] {};
    \node[state]          (34) [right of=33] {};
    \node[state]	   (f) [right of=24] {};
    \node (label) [above of= init_a, node distance = 1.8cm] {(P3)};

    \path[->] 
  (init_a) edge              node [above] {$t_1$}     (11)
  (init_a) edge              node [above] {$t_2$}     (21)
  (init_a) edge              node [above] {$t_3$}     (31)
  (11) edge              node [above] {$T,0$}     (12)
  (12) edge              node [above] {$T,0$}     (13)
  (13) edge              node [above] {$T,0$}     (14)
  (14) edge              node [above] {$\bar{v}_{t_1}$}     (f)
  (21) edge              node [above] {$T,0$}     (22)
  (22) edge              node [above] {$T,0$}     (23)
  (23) edge              node [above] {$T,0$}     (24)
  (24) edge              node [above] {$\bar{v}_{t_2}$}     (f)
  (31) edge              node [above] {$T,0$}     (32)
  (32) edge              node [above] {$T,0$}     (33)
  (33) edge              node [above] {$T,0$}     (34)
  (34) edge              node [below] {$\bar{v}_{t_3}$}     (f)
  (f) edge [loop above]             node [above] {$\Sigma$}     ();

    \node[state]  (init_b) [right of=t3, node distance=3cm] {}; 
    \node[state]          (21) [right of=init_b] {};
    \node[state]          (11) [above of=21] {};
    \node[state]          (12) [right of=11] {};
    \node[state]          (13) [right of=12] {};
    \node[state]          (14) [right of=13] {};
    \node[state]          (22) [right of=21] {};
    \node[state]          (23) [right of=22] {};
    \node[state]          (24) [right of=23] {};
    \node[state]          (31) [below of=21] {};
    \node[state]          (32) [right of=31] {};
    \node[state]          (33) [right of=32] {};
    \node[state]          (34) [right of=33] {};
    \node[state]	   (f) [right of=24] {};
    \node (label) [above of= init_b, node distance = 1.8cm] {(P2)};

    \path[->] 
  (init_b) edge              node [above] {$t_1$}     (11)
  (init_b) edge              node [above] {$t_2$}     (21)
  (init_b) edge              node [above] {$t_3$}     (31)
  (11) edge              node [above] {$T,1$}     (12)
  (12) edge              node [above] {$T,1$}     (13)
  (13) edge              node [above] {$T,1$}     (14)
  (14) edge              node [above] {$T \backslash t_1$}     (f)
  (21) edge              node [above] {$T,1$}     (22)
  (22) edge              node [above] {$T,1$}     (23)
  (23) edge              node [above] {$T,1$}     (24)
  (24) edge              node [above] {$T \backslash t_2$}     (f)
  (31) edge              node [above] {$T,1$}     (32)
  (32) edge              node [above] {$T,1$}     (33)
  (33) edge              node [above] {$T,1$}     (34)
  (34) edge              node [below] {$T \backslash t_3$}     (f)
  (f) edge [loop above]             node [above] {$\Sigma$}     ()
  (t1) edge              node [above] {$\epsilon$}     (init_a)
  (t1) edge              node [above] {$\epsilon$}     (init_b)
  (t2) edge              node [above] {$\epsilon$}     (init_a)
  (t2) edge              node [above] {$\epsilon$}     (init_b)
  (t3) edge              node [above] {$\epsilon$}     (init_a)
  (t3) edge              node [above] {$\epsilon$}     (init_b);

\end{tikzpicture}
\caption{The NBA $\babis$ from the construction in the proof of Thm.~\ref{exptime} for $m=3$ and the
tiling system $\mathcal T=(T,H,V,t_1,t_3)$ where $T=\{t_1,t_2,t_3\}$, $H=\{(t_1,t_1),(t_1,t_3),(t_3,t_3)\}$, and $V=\{(t_1,t_2),(t_2,t_1),(t_2,t_3)\}$. $\bar v_t$
denotes the set of tiles that are not vertically compatible with $t$.}
\label{fig:autexptime}
\end{figure}
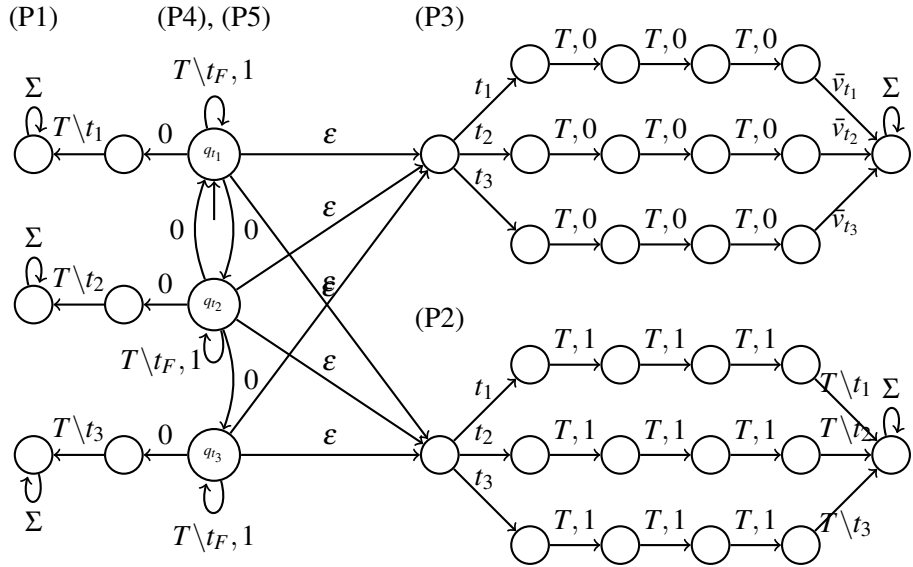

In order to establish an even higher lower bound for the continuous game we consider an EXPTIME-hard game-theoretic variant
of the tiling problem on some tiling system $\mathcal T$. The game is played by two players: \emph{Starter} and \emph{Completer}. 
The task for Completer is to produce a valid tiling, whereas Starter's goal is to 
make it impossible. On every round $i\geq 1$, 
\begin{enumerate}
\item Starter selects the tile $t_{1,i}$ starting the $i$-th row; 
if $i=1$, then $t_{1,i} = t_I$, otherwise $(t_{1,i-1},t_{1,i}) \in V$.
\item Completer selects the 
tiles $t_{2,i},\dots t_{n,i}$ completing the $i$-th row; 
$(t_{1,i},t_{2,i}), \ldots, (t_{n-1,i},t_{n,i}) \in H$,
and $(t_{2,i-1},t_{2,i}), \ldots, (t_{n,i-1},t_{n,i}) \in V$.
\end{enumerate}
If one of the players gets stuck, the opponent wins. 
Otherwise Completer wins iff there are $i,j$ such that $t_{i,j} = t_F$.
The problem of deciding whether there exists a winning strategy for Starter
in this tiling game is known to be EXPTIME-hard \cite{Chlebus86,Boas97theconvenience}. Equally, deciding whether there
is no winning strategy for him is EXPTIME-hard and -- since the games are easily seen to be determined -- so is the
problem of deciding whether or not Completer has a winning strategy. This distinction is important because, as in the
previous construction, we will reduce the complement of the tiling game problem to continuous buffered simulation. In
other words, we present a reduction from one game to another in which the players' roles are inverted. Thus, Starter in
the tiling game corresponds to Duplicator in the simulation game, and so do Completer and Spoiler.

\begin{theorem}\label{exptime}
Deciding $\sqsubseteq_{\cont}$ (resp. $\sqsubseteq^{\fair}_{\cont}$) 
is EXPTIME-hard.
\end{theorem}

\begin{proof}
Given a tiling system $\mathcal{T} = (T,H,V,t_0,t_F)$, 
we construct two NBA $\ba$, $\babis$ 
of polynomial size, that only contain accepting states,
such that there is a winning
strategy for Starter in the tiling game
if and only if there is a winning strategy for Duplicator in the 
continuous simulation game ($\ba\sqsubseteq_{\cont}\babis$, resp 
$\ba\sqsubseteq_{\cont}^{\fair}\babis$).

We consider the alphabet $T \uplus \{0,1\}$.
Spoiler's automaton $\ba$ is defined 
such that 
an infinite word $w$ is accepted by $\ba$ if and only if
it is of the form $b_0w_0b_1w_1b_2w_2\dots$, where for all $i\geq0$,
$b_i\in\{0,1\}$, $w_i\in T^n$, and
two consecutive tiles in $w_i$ are in the horizontal relation.

Duplicator's automaton does several things. It forces Spoiler
to repeat the previous row when bit $1$ occurs, i.e. if Spoiler
plays $w_i1w_{i+1}$,
then $w_i=w_{i+1}$. Duplicator also forces Spoiler
to provide a vertically matching row when bit $0$ occurs, i.e.
if Spoiler 
plays $w_i0w_{i+1}$, then $w_i$ and $w_{i+1}$ must be 
vertically compatible consecutive rows.
However, Duplicator does more: she always forces Spoiler
to start the row with a given tile $t$; this tile is determined
by the state $q_t$ 
in which Duplicator currently is.
Informally, the states $q_t$ of Duplicator's automaton $\babis$ are
such that (1) $q_{t_I}$ is the initial state of $\babis$, and 
(2) if one starts reading from $q_t$, the following holds:
\begin{enumerate} 
\item[(P1)] for
an infinite word starting with $0t'\dots$, with $t\neq t'$, 
one 
can pick an accepting run that does not depend on the infinite suffix;
\item[(P2)] for an infinite word starting with $bv1v'\dots$,
$b\in\{0,1\}$, $v,v'\in T^n$, and $v\neq v'$, 
one 
can pick an accepting run that does not depend on the infinite suffix;
\item[(P3)] for an infinite word starting with $bv0v'\dots$,
$b\in\{0,1\}$, $v,v'\in T^n$, if 
there is $i\in\{1,\dots,n\}$ such that the $i$-th letters of $v$ and $v'$ are 
not vertically compatible, then 
one 
can pick an accepting run that does not depend on the infinite suffix;
\item[(P4)] if $v\in T^n$ does not contain $t_F$, 
then $q_t\ar{1v}q_t$;
\item[(P5)] if $v\in T^n$ does not contain $t_F$, 
then $q_t\ar{0v}q_{t'}$ 
for all $t'$ such that $(t,t')\in \mathbin{V}$.
\end{enumerate}
We illustrate the construction of $\babis$ 
in Figure~\ref{fig:autexptime}.

The main component is formed by the states $q_{t_i}$ for $t_i \in T$.
Each $q_{t_i}$ is connected to another component that can detect a vertical mismatch (P3)
and a non-proper repetition (P2). Each state $q_{t_i}$ is also connected to a component that
can detect when Spoiler does not respect Duplicator's choice of the first tile (P1). 
Each state $q_{t_i}$ has a self-loop by reading $T\backslash t_F$ or $1$ (P4)
to consume the buffer and form an accepting run if one of Spoiler's mistakes is detected.
Moreover, the automaton $\babis$ encodes vertical compatibility for Duplicator's choice of the
first tile by having edges $(q_{t_i},0,q_{t_{i'}})\in \delta$ if and only if $(t_i,t_{i'}) \in V$ (P5).

We first show that if Completer has a winning strategy in the tiling game, then Spoiler has a 
winning strategy in the continuous fair simulation game on $\ba$ and $\babis$.
Spoiler plays as follows: first, he moves along 
$0v_1$, where $v_1$ is the first row of the tiling.
Then he iterates $1v_1$ for a while. 
This forces Duplicator to eventually remove $0v_1$ 
from the buffer, and commit to choosing some $q_t$, due to (P4) and (P5). 
Spoiler then considers the second row
$v_2$ that Completer would answer if Starter would put
$t$ at the beginning of the second row.  
Spoiler picks this 
row $v_2$, and plays $0v_2$, followed by iterations of $1v_2$, and repeats
the same principle.

Now we show that if Starter has a winning strategy  
then Duplicator has a winning strategy.
Duplicator first waits for the $2n+2$ first letters of Spoiler.
Because of (P1--P3), Spoiler has nothing better to do than
to play $0v1v$ for some $v$
encoding a valid first row of a tiling.
Duplicator considers the tile $t$ that would be played by Starter
in the second row if Completer played $v$ on the first row.
Duplicator then removes $0v$
and ends in the state $q_t$. From there, she waits again for $n+1$ letters,
so that the buffer now contains $1vbv'$ for some $b\in\{0,1\}$.
Repeating the same process if $b=1$, she can force Spoiler
to eventually play $0v'$ where $v'$ encodes a row vertically compatible
with $v$ and starting with $t$. Iterating this principle results
in a play won by Duplicator, since either Completer never uses the final tile or 
Spoiler's move can always be mimicked by Duplicator due to (P4) and (P5) or,
when Completer gets stuck on some row, 
Spoiler is forced to play a word with a vertical mismatch, and
Duplicator wins by accepting the rest of the word.
\end{proof}

One may wonder why the EXPTIME-hardness proof for continuous simulation does not need the machinery of the binary
counter as used in the PSPACE-hardness proof for look-ahead simulation. The reason is the following. In the look-ahead
game Duplicator always has to flush the buffer entirely. Thus, she has to wait for the entire row-by-row tiling to be produced
by Spoiler before she can point out a mistake. Thus, her best strategy is to wait for as long as possible but this would make
her lose ultimately. The integrated counter forces Spoiler to get closer and closer to the moment when he has to play the
final tile, and Duplicator can therefore relax and wait for that moment before she flushes the entire buffer. In the continuous
game, Duplicator's ability to consume parts of the buffer is enough to force Spoiler to not delay the production of a proper
tiling forever.


%% file: quotient.tex
\label{sec:quotient}
We now show that the bounds of the previous section are tight by establishing the decidability 
of buffered simulations with corresponding complexity bounds. For this, we define a ``quotient
 game'' that has a finite
state space, and show that it is equivalent to the buffered simulation
game. 

\paragraph*{Continuous Quotient Game.}
The quotient game is based on the congruence relation
associated with the Ramsey-based algorithm for complementation.
We briefly recall its definition.
Let us fix two B\"uchi automata $\ba=(Q,\Sigma,\delta,q_I,F)$ and 
$\babis=(Q,\Sigma,\delta,q_I',F)$ -- for simplicity we assume they share
the same state space and only differ in their initial state.  
We introduce the function 
$f_w: Q^2 \to \{0,1,2\}$ defined as 
$$
f_w(q,q')= \left\{\begin{array}{l@{\quad}l}
0 & \mbox{if }q\aar{w}q'\\
1 & \mbox{if }q\not\ar{w}q'\\
2 & \mbox{otherwise}
\end{array}\right.
$$
We say that two finite words $w,w'\in\Sigma^{*}$ are equivalent, 
$w\sim w'$, if $f_w=f_{w'}$. Observe that $\sim$ is an equivalence
relation, a congruence for word concatenation, and that 
the number $|\Sigma^*/{\sim}|$ of equivalence classes is bounded by $3^{|Q|^2}$.
We write $[w]$ to denote the equivalence class of $w$ with respect to $\sim$. 
We say a class $[w]$ is idempotent if $[ww]=[w]$.

\begin{definition} \rm
\label{def:continuous-quotient-game}
The \emph{continuous quotient game} is played between players Refuter and Prover\footnote{We use different player names on purpose to make an easy distinction between the original
  simulation game and the quotient game.}
as follows. Initially, Refuter's pebble is on $q_0:=q_I$ and Prover's pebble is on 
$q_0':=q_I'$. The players use an abstraction by equivalence classes of a buffer that, initially, 
contains $[\epsilon]$. On each round $i \geq 1$:
\begin{enumerate}
 \item Refuter chooses two equivalence classes $[w_1]$, $[w_2]$ and
a state $q_i$, such that $q_{i-1}\ar{w_1}q_i\aar{w_2}q_i$
 and $[w_2]$ is idempotent
 \item Prover chooses $q_i'$ such that  $q_{i-1}'\xrightarrow{\beta w_1 w_2} q_{i}'\aar{w_2}q_i'$.
 The value $\beta$ of the abstract buffer is set to $[w_2]$ for the next turn.
\end{enumerate}
Prover wins the play if Refuter gets stuck or the play is 
infinitely long.
\end{definition}

\begin{proposition}
\label{prop:continuous-decidable}
Whether Prover has a winning strategy for the
continuous quotient game is decidable in EXPTIME.
\end{proposition}
\begin{proof}
Observe first that the arena of the quotient game is finite
and can be computed in exponential time.
Indeed, a configuration of a quotient game is either a tuple 
$(q,q',[w])$ for Refuter's turn or a tuple $(q,q',[b],[w],[w'])$ for Prover's turn.
The arena of the quotient game is thus finite and its size is bounded by 
$2|Q|^2\cdot |\Sigma^*/{\sim}|^2=2|Q|^2\cdot 3^{2|Q|^2} = 2^{\mathcal{O}(|Q|^2 \cdot \log |Q|)}$. The finite monoid
$\Sigma^*/{\sim}$ can be computed in exponential time: starting from the set 
$\{ [a] \mid a \in \Sigma \}$, compose any two classes until a fixpoint is reached. Composition 
of two equivalence classes given as functions of type $Q^2 \to \{0,1,2\}$ is not hard to compute
\cite{dhl-fsttcs06}.

Observe now that the quotient game is a reachability game
from Refuter's point of view (he wins if he reaches a configuration in which
Prover gets stuck), so once the arena is computed, one can decide the
winner of the game in time polynomial in the size of the arena, which
is exponential in $|Q|$.
\end{proof}


We show that quotient games characterise the relation $\sqsubseteq^{\fair}_{\cont}$. 

\begin{lemma}
\label{lem:quotient2continuous}
$\ba\sqsubseteq^{\fair}_{\cont}\babis$ only if Prover has a winning
strategy for the continuous quotient game.
\end{lemma}

\begin{proof}
Assume that Refuter has a winning strategy for the continuous
quotient game.
We want to show that then Spoiler has a winning strategy for the continuous
fair simulation game. We actually consider a variant of the continuous
fair simulation game in which Spoiler may add more than one letter in a 
round, and Duplicator only removes one letter in a round. Clearly, Spoiler has
a winning strategy for this variant if and only if he has a winning strategy
for the continuous fair simulation game as defined in 
Section~\ref{sec:continuous}. Spoiler's strategy basically follows
the one of Refuter. In the first round, Spoiler adds into the buffer
some representatives $w_1,w_2$ of the equivalence classes played by
Refuter. Spoiler then adds $w_2$ into the buffer on every round 
until the answer of Duplicator can be identified as a Prover's move
in the quotient game, i.e.\ if Duplicator does not get stuck, she will eventually produce
a trace of the form
$q'_0 \xrightarrow{w_1w_2^*} q'_1 \aar{w_2^+} q'_1$, since
there are only finitely many states in the automaton.
Then Spoiler considers the state $q_1'$ in which Duplicator
is and looks at what Refuter would play if Prover would have
picked $q_1'$. Iterating this principle, Spoiler mimics
Refuter's winning strategy: eventually, since Prover gets stuck on some round $i$,
Duplicator will get stuck when trying to mimick $w_1w_2^+\ldots w_i^+$, and then
Spoiler wins by continuously adding $w_i$ into the buffer for the rest of the play.
\end{proof}

A key argument in the proof of the converse direction is the following lemma which
is easily proved using Ramsey's Theorem \cite{Ramsey:30}.

\begin{lemma}
\label{lem:ramsey}
Let $q_0,q_1,\dots$ be an infinite accepting run on $a_1 a_2 \ldots$. Then there are
$i,j,k$ with $i<j<k$ such that $q_i=q_j=q_k$ is accepting 
and $a_{i+1}\dots a_{j}\sim a_{j+1}\dots a_{k}\sim a_{i+1}\dots a_{k}$.
\end{lemma}

\begin{lemma}
\label{lem:continuous2quotient}
$\ba\sqsubseteq^{\fair}_{\cont}\babis$ if Prover has a winning strategy for the continuous 
quotient game.
\end{lemma}

\begin{proof}
When the continuous simulation game starts, Duplicator just
skips his turn for a while. Then Spoiler starts providing
an infinite accepting run $q_0 a_0 q_1 a_1\dots$ -- if he does not, Duplicator
waits forever and wins the play. At some point, 
Lemma~\ref{lem:ramsey} applies: 
the buffer contains $w_1w_2w_2'$ with $[w_2]=[w_2']$ being idempotent,
and Spoiler is in a state $q$ that admits a $[w_2]$-loop. Then Duplicator
considers the state $q'$ in which Prover would move 
if Refuter played $[w_1],[w_2],q$ in the first round. 
She removes $w_1w_2$ from the buffer and moves to this state $q'$.
Duplicator proceeds identically in the next rounds, and either Spoiler
eventually gets stuck or he follows a non-accepting run or the play
is infinite.
\end{proof}

Lemmas~\ref{lem:quotient2continuous} and \ref{lem:continuous2quotient} together with
Prop.~\ref{prop:continuous-decidable} yield an upper bound on the complexity of deciding
continuous simulation. Together with the lower bound from Theorem~\ref{exptime} we get 
a complete characterisation of the complexity of continuous fair simulation.

\begin{corollary}
Continuous fair simulation is EXPTIME-complete. 
\end{corollary}

\paragraph*{Look-Ahead Quotient Game.}
In order to establish the decidability of look-ahead simulations, we 
introduce a look-ahead quotient game. The game essentially
differs from the continuous quotient game in that it does not use
a buffer.

\begin{definition}
\label{def:look-quotient-game}
The look-ahead quotient game is played between Refuter and Prover. Initially, Refuter's 
pebble is on $q_0:=q_I$, Prover's pebble is on $q_0':=q_I'$, and the buffer 
$\beta$ contains the equivalence class $[\epsilon]$.
On each round $i \geq 1$:
\begin{enumerate}
 \item Refuter chooses two equivalence classes $[w_1]$, $[w_2]$ and
a state $q_i$, such that $q_{i-1}\ar{w_1}q_i\aar{w_2}q_i$ and
$[w_2]$ is idempotent.
 \item Prover chooses $q_i'$ such that there is a 
$q_{i-1}'\ar{w_1}q_{i}'\aar{w_2}q_i'$.
\end{enumerate}
Prover wins the play if Refuter gets stuck or if the play is 
infinitely long.
\end{definition}

Following the same kind of arguments we used for the continuous quotient
game, the result below can be established.

\begin{proposition}
$\ba\sqsubseteq^{\fair}_{\look}\babis$ if and only if Prover has a winning
strategy for the look-ahead quotient game.
\end{proposition}

The size of the arena of a look-ahead quotient game is again
exponential in the size of the automata; but there are
only $|Q|^2$ positions for Refuter, so look-ahead quotient
games can be solved slightly better than continuous ones.

\begin{proposition}
\label{lem:look-decidable}
Whether Prover has a winning strategy for 
the look-ahead quotient game can be decided in PSPACE.
\end{proposition}

\begin{proof}
Consider the following non-deterministic algorithm that guesses the set
$W$ of all pairs $(q_0,q_0')$ of initial configurations of the game 
such that Duplicator has a winning strategy. 
For all $(q_0,q_0')$ in $W$, the following can then be checked
in polynomial space:
for all $[w_1],[w_2]$, and $q_1$ that could be played by Spoiler,
there is $q_1'$ that can be played by Duplicator such that $(q_1,q_1')$ is in 
$W$. Inclusion in PSPACE then follows from Savitch's Theorem \cite{Savitch70}.
\end{proof}

\begin{corollary}
Look-ahead fair simulation is PSPACE-complete.
\end{corollary}


%% file: properties.tex
In this section we investigate some fundamental properties of buffered simulations starting
with a comparison to language inclusion. Remember that the main motivation for studying
simulations is the approximation thereof.

\paragraph*{Continuous Simulation vs.\ Language Inclusion.}
Continuous simulation is strictly smaller than language inclusion. It is not hard to see that
continuous simulation implies language inclusion, so we focus on strictness.

The following example shows a case where language inclusion holds, indeed $L(\ba)=L(\babis)$,
but $\ba \not \sqsubseteq^{\fair}_{\cont} \babis$ since Spoiler
can win the game by always producing $a$, whereas Duplicator has to keep the pebble on the initial
state of $\babis$ to be ready for a possible $b$.

\begin{center}
\begin{tikzpicture}[initial text={}, node distance=12mm, every state/.style={minimum size=5mm}, thick]
  \node[state,initial,accepting]   (0) 		    {};
  \node[state,accepting]           (1) [right of=0] {};
  \node[state,initial]   (2)  [right of=1, node distance=5cm] {};
  \node[state,accepting] (3) [above right of=2] {};
  \node[state,accepting] (4) [below right of=3] {};

  \path[->] (0) edge[loop above] node [above] {$a$} (0)
                edge node [above] {$b$}              (1)

            (1) edge[loop above] node  {$b$}        (1)

            (2) edge node [above] {$a$}              (3)
                edge node [below] {$b$}              (4)
                edge[loop above] node [above] {$a$}  (2)

            (3) edge[loop right] node {$a$} (3)
            (4) edge[loop right] node {$b$} (4);
\end{tikzpicture}
\end{center}

\paragraph*{Topological Characterisation.}
Consider a run of an NBA on some word $w = a_1a_2 \ldots \in \Sigma^\omega$ to be an 
infinite sequence $q_0,a_1,q_1,\ldots$ with the usual properties, i.e.\ the word is actually listed 
in the run itself. We write $\Runs(\ba)$ for the set of runs of $\ba$ in this respect, and $\ARuns(\ba)$ for the set of accepting runs. 

Given a set $\Delta$, the set $\Delta^{\omega}$ is equipped
with a standard structure of a metric space. The distance
$d(x,y)$ between two infinite sequences $x_0x_1x_2\dots$ 
and $y_0y_1y_2\dots$ is the real $\frac{1}{2^{i}}$, where $i$ is the first
index for which $x_i\neq y_i$. Intuitively,
two words are ``significantly close'' if they share a ``significantly long'' 
prefix. 
The sets $\Runs(\ba)$ and $\ARuns(\ba)$ are subsets of 
$(Q\cup\Sigma)^{\omega}$; $\Runs(\ba)$ has the particularity
of being a closed subset, and it is thus a compact space, whereas
$\ARuns(\ba)$ is not.

We call a function $f:\ARuns(\ba)\to\ARuns(\babis)$ \emph{word preserving}
if for all $\rho\in\ARuns(\ba)$, $f(\rho)$ and $\rho$ are 
labelled with the same word.
It can be seen that $L(\ba)\subseteq L(\babis)$
holds if and only if there is a word preserving function 
$f:\ARuns(\ba)\to\ARuns(\babis)$.

\begin{proposition}
\label{prop:topology}
Let $\ba,\babis$ be two NBA. 
The following holds:
$\ba\sqsubseteq^{\fair}_{\cont}\babis$ if and only if there is a continuous
word preserving function
$f:\ARuns(\ba)\to\ARuns(\babis)$.
\end{proposition}


Proposition~\ref{prop:topology} 
has some interesting consequences. First, it shows
again that $\ba\sqsubseteq_{\cont}^{\fair}\babis$ implies 
$L(\ba)\subseteq L(\babis)$, and explain the difference between the two in 
terms of continuity. Second, it shows that $\sqsubseteq_{\cont}$ and
$\sqsubseteq^\fair_{\cont}$ are transitive relations, since the composition
of two continuous functions is continuous. 
Another application of Proposition~\ref{prop:topology}
is that $\sqsubseteq_{\cont}$ (but not $\sqsubseteq^{\fair}_{\cont}$)
is decidable in 2-EXPTIME using 
a result of Holtmann et al.~\cite{DBLP:journals/corr/abs-1209-0800}. This is of course not
optimal as seen in the previous section.

\begin{remark}
It might be asked whether look-ahead simulation has a topological 
characterisation similar to this one. The answer is negative:
if it had (a reasonable) one, it would entail that look-ahead simulation is
a transitive relation. However,
Mayr and Clemente~\cite{MayrC13} gave examples of automata that 
show that look-ahead simulation is not transitive in general.
\end{remark}


%% file: conclusion.tex
An important application of simulation relations in automata theory is automata 
minimisation. A preoder $R$ over the set of states of an automaton $\ba$ 
defines two new automata: its quotient $\ba/R$, and its pruning 
$\mathsf{prune}(\ba,R)$, c.f.\ Clemente's PhD thesis~\cite{phd-clemente} for a formal 
definition of these notions. Intuitively, the 
quotient automaton is defined by merging states that are equivalent with 
respect
to the preorder $R$, whereas pruning is obtained by removing a transition
$q\ar{a}q_1$ if it is ``subsumed'' by a transition $q\ar{a}q_2$, where 
$q_1 \mathbin{R} q_2$. 

A preoder $R$ is then said to be good for quotienting (GFQ)
if $L(\ba/R)=L(\ba)$, and good for pruning (GFP) if 
$L(\mathsf{prune}(\ba,R))=L(\ba)$. It can be checked that GFQ and GFP
are antitone properties: if $R\supseteq R'$ and $R$ is GFQ (resp. GFP),
then so does $R'$.

Fair simulation is neither GFQ nor GFP; as a consequence, fair continuous
and fair look-ahead simulations, which contain fair simulation, are not
GFQ and GFP either. Simulation preorders that are used for automata
minimisation rely on less permissive winning conditions than fairness.
The \emph{delayed} winning condition asserts that every round in which
Spoiler visits an accepting state is (not necessarily immediately) succeeded by some
round in which Duplicator also visits an accepting state. The \emph{direct} winning 
condition imposes that, if 
Spoiler visits an accepting state in a given round, then in the same round
Duplicator should visit an accepting state. Delayed simulation
is known to be GFQ but not GFP, whereas direct simulation is known to be both
GFP and GFQ. Since a play of a continuous/look-ahead simulation game yields
a play of the standard simulation game, there is a natural buffered counterpart
of delayed and direct simulation, obtained by changing the winning conditions accordingly.


\begin{proposition}
\label{prop:sound}
Delayed continuous and delayed look-ahead simulation is GFQ but not GFP, and 
direct continuous as well as direct look-ahead simulation is GFP and GFQ.
\end{proposition}

The proof is a rather straightforward consequence of similar results for
multi-pebble simulations~\cite{pebble}, and from the fact that these
multi-pebble simulations subsume continuous simulations (provided the
number of pebbles is larger than the number of states of duplicator's 
automaton). 

Recall that bounded buffered simulation relations are polynomial time computable
\cite{HutagalungLL13} and can be used to significantly improve 
language inclusion
tests for NBA using automata minimisation \cite{MayrC13}. 
We already showed that fair, unbounded, buffered simulation is not
polynomial time computable, and thus cannot be used for improving language 
inclusion tests. We now extend this result to the delayed and direct buffered
simulations.

\begin{theorem}
The delayed (resp. direct) continuous simulation is EXPTIME hard, and
the delayed (resp. direct) look-ahead simulation is PSPACE hard.
\end{theorem}

This follows from a simple observation:
the automata that were used in
the hardness proofs had all states accepting, and in this case, 
fair, delayed and direct simulation coincide.


%% file: main.bbl
\begin{thebibliography}{10}
\providecommand{\bibitemdeclare}[2]{}
\providecommand{\surnamestart}{}
\providecommand{\surnameend}{}
\providecommand{\urlprefix}{Available at }
\providecommand{\url}[1]{\texttt{#1}}
\providecommand{\href}[2]{\texttt{#2}}
\providecommand{\urlalt}[2]{\href{#1}{#2}}
\providecommand{\doi}[1]{doi:\urlalt{http://dx.doi.org/#1}{#1}}
\providecommand{\bibinfo}[2]{#2}

\bibitemdeclare{inproceedings}{AbdullaCCHHMV10}
\bibitem{AbdullaCCHHMV10}
\bibinfo{author}{P.~A. \surnamestart Abdulla\surnameend},
  \bibinfo{author}{Y.-F. \surnamestart Chen\surnameend},
  \bibinfo{author}{L.~\surnamestart Clemente\surnameend},
  \bibinfo{author}{L.~\surnamestart Hol\'{\i }k\surnameend},
  \bibinfo{author}{C.-D. \surnamestart Hong\surnameend},
  \bibinfo{author}{R.~\surnamestart Mayr\surnameend} \&
  \bibinfo{author}{T.~\surnamestart Vojnar\surnameend} (\bibinfo{year}{2010}):
  \emph{\bibinfo{title}{Simulation Subsumption in Ramsey-Based {B}{\"u}chi
  Automata Universality and Inclusion Testing}}.
\newblock In: {\sl \bibinfo{booktitle}{Proc.\ 22nd Int.\ Conf.\ on
  Computer-Aided Verification, {CAV'10}}}, {\sl \bibinfo{series}{LNCS}}
  \bibinfo{volume}{6174}, \bibinfo{publisher}{Springer}, pp.
  \bibinfo{pages}{132--147}, \doi{10.1007/978-3-642-14295-6\_14}.

\bibitemdeclare{inproceedings}{AbdullaCCHHMV11}
\bibitem{AbdullaCCHHMV11}
\bibinfo{author}{P.~Aziz \surnamestart Abdulla\surnameend},
  \bibinfo{author}{Y.-F. \surnamestart Chen\surnameend},
  \bibinfo{author}{L.~\surnamestart Clemente\surnameend},
  \bibinfo{author}{L.~\surnamestart Hol\'{\i }k\surnameend},
  \bibinfo{author}{C.-D. \surnamestart Hong\surnameend},
  \bibinfo{author}{R.~\surnamestart Mayr\surnameend} \&
  \bibinfo{author}{T.~\surnamestart Vojnar\surnameend} (\bibinfo{year}{2011}):
  \emph{\bibinfo{title}{Advanced Ramsey-Based {B}{\"u}chi Automata Inclusion
  Testing}}.
\newblock In: {\sl \bibinfo{booktitle}{Proc.\ 22nd Int.\ Conf.\ on Concurrency
  Theory, {CONCUR'11}}}, {\sl \bibinfo{series}{LNCS}} \bibinfo{volume}{6901},
  \bibinfo{publisher}{Springer}, pp. \bibinfo{pages}{187--202},
  \doi{10.1007/978-3-642-23217-6\_13}.

\bibitemdeclare{inproceedings}{Boas97theconvenience}
\bibitem{Boas97theconvenience}
\bibinfo{author}{Peter Van~Emde \surnamestart Boas\surnameend}
  (\bibinfo{year}{1997}): \emph{\bibinfo{title}{The Convenience of Tilings}}.
\newblock In: {\sl \bibinfo{booktitle}{In Complexity, Logic, and Recursion
  Theory}}, \bibinfo{publisher}{Marcel Dekker Inc}, pp.
  \bibinfo{pages}{331--363}, \doi{10.1.1.38.763}.

\bibitemdeclare{inproceedings}{buc62}
\bibitem{buc62}
\bibinfo{author}{J.~R. \surnamestart B{\"u}chi\surnameend}
  (\bibinfo{year}{1962}): \emph{\bibinfo{title}{On a Decision Method in
  Restricted Second Order Arithmetic}}.
\newblock In: {\sl \bibinfo{booktitle}{Proc.\ Congress on Logic, Method, and
  Philosophy of Science}}, \bibinfo{publisher}{Stanford University Press},
  \bibinfo{address}{Stanford, CA, USA}, pp. \bibinfo{pages}{1--12},
  \doi{10.1007/978-1-4613-8928-6\_23}.

\bibitemdeclare{article}{Bustan:2003:SM:635499.635502}
\bibitem{Bustan:2003:SM:635499.635502}
\bibinfo{author}{D.~\surnamestart Bustan\surnameend} \&
  \bibinfo{author}{O.~\surnamestart Grumberg\surnameend}
  (\bibinfo{year}{2003}): \emph{\bibinfo{title}{Simulation-based
  minimization}}.
\newblock {\sl \bibinfo{journal}{ACM Trans. Comput. Logic}}
  \bibinfo{volume}{4}(\bibinfo{number}{2}), pp. \bibinfo{pages}{181--206},
  \doi{10.1145/635499.635502}.

\bibitemdeclare{article}{CeceF05}
\bibitem{CeceF05}
\bibinfo{author}{G.~\surnamestart C{\'e}c{\'e}\surnameend} \&
  \bibinfo{author}{A.~\surnamestart Finkel\surnameend} (\bibinfo{year}{2005}):
  \emph{\bibinfo{title}{Verification of programs with half-duplex
  communication}}.
\newblock {\sl \bibinfo{journal}{Inf. Comput.}}
  \bibinfo{volume}{202}(\bibinfo{number}{2}), pp. \bibinfo{pages}{166--190},
  \doi{10.1016/j.ic.2005.05.006}.

\bibitemdeclare{article}{Chlebus86}
\bibitem{Chlebus86}
\bibinfo{author}{B.~S. \surnamestart Chlebus\surnameend}
  (\bibinfo{year}{1986}): \emph{\bibinfo{title}{Domino-Tiling Games}}.
\newblock {\sl \bibinfo{journal}{Journal of Computer and System Sciences}}
  \bibinfo{volume}{32}, pp. \bibinfo{pages}{374--392},
  \doi{10.1016/0022-0000(86)90036-X}.

\bibitemdeclare{phdthesis}{phd-clemente}
\bibitem{phd-clemente}
\bibinfo{author}{L.~\surnamestart Clemente\surnameend} (\bibinfo{year}{2012}):
  \emph{\bibinfo{title}{Generalized Simulation Relations with Applications in
  Automata Theory}}.
\newblock Ph.D. thesis, \bibinfo{school}{University of Edinburgh}.

\bibitemdeclare{inproceedings}{MayrC13}
\bibitem{MayrC13}
\bibinfo{author}{Lorenzo \surnamestart Clemente\surnameend} \&
  \bibinfo{author}{Richard \surnamestart Mayr\surnameend}
  (\bibinfo{year}{2013}): \emph{\bibinfo{title}{Advanced automata
  minimization}}.
\newblock In: {\sl \bibinfo{booktitle}{Proc.\ 40th Symp.\ on Principles of
  Programming Languages, {POPL'13}}}, \bibinfo{publisher}{ACM}, pp.
  \bibinfo{pages}{63--74}, \doi{10.1145/2429069.2429079}.

\bibitemdeclare{inproceedings}{dhl-fsttcs06}
\bibitem{dhl-fsttcs06}
\bibinfo{author}{C.~\surnamestart Dax\surnameend},
  \bibinfo{author}{M.~\surnamestart Hofmann\surnameend} \&
  \bibinfo{author}{M.~\surnamestart Lange\surnameend} (\bibinfo{year}{2006}):
  \emph{\bibinfo{title}{A proof system for the linear time $\mu $-calculus}}.
\newblock In: {\sl \bibinfo{booktitle}{Proc.\ 26th Conf.\ on Foundations of
  Software Technology and Theoretical Computer Science, {FSTTCS'06}}}, {\sl
  \bibinfo{series}{LNCS}} \bibinfo{volume}{4337},
  \bibinfo{publisher}{Springer}, pp. \bibinfo{pages}{274--285},
  \doi{10.1007/11944836\_26}.

\bibitemdeclare{inproceedings}{DillHW91}
\bibitem{DillHW91}
\bibinfo{author}{D.~L. \surnamestart Dill\surnameend}, \bibinfo{author}{A.~J.
  \surnamestart Hu\surnameend} \& \bibinfo{author}{H.~\surnamestart
  Wong-Toi\surnameend} (\bibinfo{year}{1991}): \emph{\bibinfo{title}{Checking
  for Language Inclusion Using Simulation Preorders}}.
\newblock In: {\sl \bibinfo{booktitle}{Proc.\ 3rd Int.\ Workshop on
  Computer-Aided Verification, {CAV'91}}}, {\sl \bibinfo{series}{LNCS}}
  \bibinfo{volume}{575}, \bibinfo{publisher}{Springer}, pp.
  \bibinfo{pages}{255--265}, \doi{10.1007/3-540-55179-4\_25}.

\bibitemdeclare{article}{antichain-journal}
\bibitem{antichain-journal}
\bibinfo{author}{L.~\surnamestart Doyen\surnameend} \& \bibinfo{author}{J.-F.
  \surnamestart Raskin\surnameend} (\bibinfo{year}{2009}):
  \emph{\bibinfo{title}{Antichains for the Automata-Based Approach to
  Model-Checking}}.
\newblock {\sl \bibinfo{journal}{Logical Methods in Computer Science}}
  \bibinfo{volume}{5}(\bibinfo{number}{1}), \doi{10.2168/LMCS-5(1:5)2009}.

\bibitemdeclare{inproceedings}{EtessamiH00}
\bibitem{EtessamiH00}
\bibinfo{author}{K.~\surnamestart Etessami\surnameend} \&
  \bibinfo{author}{G.~J. \surnamestart Holzmann\surnameend}
  (\bibinfo{year}{2000}): \emph{\bibinfo{title}{Optimizing {B}{\"u}chi
  Automata}}.
\newblock In: {\sl \bibinfo{booktitle}{Proc.\ 11th Int.\ Conf. on Concurrency
  Theory, {CONCUR'00}}}, {\sl \bibinfo{series}{LNCS}} \bibinfo{volume}{1877},
  \bibinfo{publisher}{Springer}, pp. \bibinfo{pages}{153--167},
  \doi{10.1007/3-540-44618-4\_13}.

\bibitemdeclare{inproceedings}{EtessamiWS01}
\bibitem{EtessamiWS01}
\bibinfo{author}{K.~\surnamestart Etessami\surnameend},
  \bibinfo{author}{T.~\surnamestart Wilke\surnameend} \& \bibinfo{author}{R.~A.
  \surnamestart Schuller\surnameend} (\bibinfo{year}{2001}):
  \emph{\bibinfo{title}{Fair Simulation Relations, Parity Games, and State
  Space Reduction for {B}{\"u}chi Automata}}.
\newblock In: {\sl \bibinfo{booktitle}{Proc.\ 28th Int.\ Coll.\ on Algorithms,
  Languages and Programming, {ICALP'01}}}, {\sl \bibinfo{series}{LNCS}}
  \bibinfo{volume}{2076}, \bibinfo{publisher}{Springer}, pp.
  \bibinfo{pages}{694--707}, \doi{10.1137/S0097539703420675}.

\bibitemdeclare{inproceedings}{pebble}
\bibitem{pebble}
\bibinfo{author}{Kousha \surnamestart Etessami\surnameend}
  (\bibinfo{year}{2002}): \emph{\bibinfo{title}{A Hierarchy of Polynomial-Time
  Computable Simulations for Automata}}.
\newblock In \bibinfo{editor}{LuboÅ¡ \surnamestart Brim\surnameend},
  \bibinfo{editor}{MojmÃ­r \surnamestart KÅ™etÃ­nskÃ½\surnameend},
  \bibinfo{editor}{AntonÃ­n \surnamestart KuÄera\surnameend} \&
  \bibinfo{editor}{Petr \surnamestart JanÄar\surnameend}, editors: {\sl
  \bibinfo{booktitle}{CONCUR 2002 â€” Concurrency Theory}}, {\sl
  \bibinfo{series}{Lecture Notes in Computer Science}} \bibinfo{volume}{2421},
  \bibinfo{publisher}{Springer Berlin Heidelberg}, pp.
  \bibinfo{pages}{131--144}, \doi{10.1007/3-540-45694-5\_10}.

\bibitemdeclare{inproceedings}{conf/tacas/FogartyV10}
\bibitem{conf/tacas/FogartyV10}
\bibinfo{author}{S.~\surnamestart Fogarty\surnameend} \& \bibinfo{author}{M.~Y.
  \surnamestart Vardi\surnameend} (\bibinfo{year}{2010}):
  \emph{\bibinfo{title}{Efficient {B}{\"u}chi Universality Checking}}.
\newblock In: {\sl \bibinfo{booktitle}{Proc.\ 16th Int.\ Conf.\ on Tools and
  Algorithms for the Construction and Analysis of Systems, {TACAS'10}}}, {\sl
  \bibinfo{series}{LNCS}} \bibinfo{volume}{6015},
  \bibinfo{publisher}{Springer}, pp. \bibinfo{pages}{205--220},
  \doi{10.1007/978-3-642-12002-2\_17}.

\bibitemdeclare{article}{fogarty-vardi-2012}
\bibitem{fogarty-vardi-2012}
\bibinfo{author}{S.~\surnamestart Fogarty\surnameend} \& \bibinfo{author}{M.~Y.
  \surnamestart Vardi\surnameend} (\bibinfo{year}{2012}):
  \emph{\bibinfo{title}{{B}{\"u}chi Complementation and Size-Change
  Termination}}.
\newblock {\sl \bibinfo{journal}{Logical Methods in Computer Science}}
  \bibinfo{volume}{8}(\bibinfo{number}{1}), \doi{10.2168/LMCS-8(1:13)2012}.

\bibitemdeclare{inproceedings}{fairsimmini}
\bibitem{fairsimmini}
\bibinfo{author}{S.~\surnamestart Gurumurthy\surnameend},
  \bibinfo{author}{R.~\surnamestart Bloem\surnameend} \&
  \bibinfo{author}{F.~\surnamestart Somenzi\surnameend} (\bibinfo{year}{2002}):
  \emph{\bibinfo{title}{Fair simulation minimization}}.
\newblock In: {\sl \bibinfo{booktitle}{Proc.\ 14th Int.\ Conf.\ on
  Computer-Aided Verification, {CAV'02}}}, {\sl \bibinfo{series}{LNCS}}
  \bibinfo{volume}{2404}, \bibinfo{publisher}{Springer}, pp.
  \bibinfo{pages}{610--624}, \doi{10.1007/3-540-45657-0\_51}.

\bibitemdeclare{article}{HenzingerKR02}
\bibitem{HenzingerKR02}
\bibinfo{author}{T.~A. \surnamestart Henzinger\surnameend},
  \bibinfo{author}{O.~\surnamestart Kupferman\surnameend} \&
  \bibinfo{author}{S.~K. \surnamestart Rajamani\surnameend}
  (\bibinfo{year}{2002}): \emph{\bibinfo{title}{Fair Simulation}}.
\newblock {\sl \bibinfo{journal}{Inf. Comput.}}
  \bibinfo{volume}{173}(\bibinfo{number}{1}), pp. \bibinfo{pages}{64--81},
  \doi{10.1006/inco.2001.3085}.

\bibitemdeclare{article}{DBLP:journals/corr/abs-1209-0800}
\bibitem{DBLP:journals/corr/abs-1209-0800}
\bibinfo{author}{M.~\surnamestart Holtmann\surnameend},
  \bibinfo{author}{L.~\surnamestart Kaiser\surnameend} \&
  \bibinfo{author}{W.~\surnamestart Thomas\surnameend} (\bibinfo{year}{2012}):
  \emph{\bibinfo{title}{Degrees of Lookahead in Regular Infinite Games}}.
\newblock {\sl \bibinfo{journal}{Logical Methods in Computer Science}}
  \bibinfo{volume}{8}(\bibinfo{number}{3}), \doi{10.2168/LMCS-8(3:24)2012}.

\bibitemdeclare{book}{spin-book}
\bibitem{spin-book}
\bibinfo{author}{Gerard~J. \surnamestart Holzmann\surnameend}
  (\bibinfo{year}{2004}): \emph{\bibinfo{title}{The SPIN Model Checker - primer
  and reference manual}}.
\newblock \bibinfo{publisher}{Addison-Wesley}.

\bibitemdeclare{inproceedings}{HutagalungLL13}
\bibitem{HutagalungLL13}
\bibinfo{author}{M.~\surnamestart Hutagalung\surnameend},
  \bibinfo{author}{M.~\surnamestart Lange\surnameend} \&
  \bibinfo{author}{{\'E}.~\surnamestart Lozes\surnameend}
  (\bibinfo{year}{2013}): \emph{\bibinfo{title}{Revealing vs. Concealing: More
  Simulation Games for {B}{\"u}chi Inclusion}}.
\newblock In: {\sl \bibinfo{booktitle}{Proc.\ 7th Int.\ Conf.\ on Language and
  Automata Theory and Applications, {LATA'13}}}, \bibinfo{series}{LNCS},
  \bibinfo{publisher}{Springer}, pp. \bibinfo{pages}{347--358},
  \doi{10.1007/978-3-642-37064-9\_31}.

\bibitemdeclare{article}{KupfermanV01}
\bibitem{KupfermanV01}
\bibinfo{author}{O.~\surnamestart Kupferman\surnameend} \&
  \bibinfo{author}{M.~Y. \surnamestart Vardi\surnameend}
  (\bibinfo{year}{2001}): \emph{\bibinfo{title}{Weak Alternating Automata Are
  Not That Weak}}.
\newblock {\sl \bibinfo{journal}{ACM Trans.\ on Comput.\ Logic}}
  \bibinfo{volume}{2}(\bibinfo{number}{3}), pp. \bibinfo{pages}{408--429},
  \doi{10.1145/377978.377993}.

\bibitemdeclare{inproceedings}{size-change-termination}
\bibitem{size-change-termination}
\bibinfo{author}{C.~S. \surnamestart Lee\surnameend}, \bibinfo{author}{N.~D.
  \surnamestart Jones\surnameend} \& \bibinfo{author}{A.~M. \surnamestart
  Ben-Amram\surnameend} (\bibinfo{year}{2001}): \emph{\bibinfo{title}{The
  size-change principle for program termination}}.
\newblock In: {\sl \bibinfo{booktitle}{Proc.\ 28th Symp.\ on Principles of
  Programming Languages, {POPL'01}}}, \bibinfo{publisher}{ACM}, pp.
  \bibinfo{pages}{81--92}, \doi{10.1145/360204.360210}.

\bibitemdeclare{inproceedings}{STOC73*1}
\bibitem{STOC73*1}
\bibinfo{author}{A.~R. \surnamestart Meyer\surnameend} \&
  \bibinfo{author}{L.~J. \surnamestart Stockmeyer\surnameend}
  (\bibinfo{year}{1973}): \emph{\bibinfo{title}{Word problems requiring
  exponential time}}.
\newblock In: {\sl \bibinfo{booktitle}{Proc.\ 5th Symp.\ on Theory of
  Computing, {STOC'73}}}, \bibinfo{publisher}{ACM}, \bibinfo{address}{New
  York}, pp. \bibinfo{pages}{1--9}, \doi{10.1145/800125.804029}.

\bibitemdeclare{article}{Ramsey:30}
\bibitem{Ramsey:30}
\bibinfo{author}{F.~P. \surnamestart Ramsey\surnameend} (\bibinfo{year}{1930}):
  \emph{\bibinfo{title}{On a problem in formal logic}}.
\newblock {\sl \bibinfo{journal}{Proc.\ London Math.\ Soc.\ (3)}}
  \bibinfo{volume}{30}, pp. \bibinfo{pages}{264--286},
  \doi{10.1007/978-0-8176-4842-8\_1}.

\bibitemdeclare{inproceedings}{FOCS88*319}
\bibitem{FOCS88*319}
\bibinfo{author}{S.~\surnamestart Safra\surnameend} (\bibinfo{year}{1988}):
  \emph{\bibinfo{title}{On the complexity of $\omega $-automata}}.
\newblock In: {\sl \bibinfo{booktitle}{Proc.\ 29th Symp.\ on Foundations of
  Computer Science, {FOCS'88}}}, \bibinfo{publisher}{IEEE}, pp.
  \bibinfo{pages}{319--327}, \doi{10.1109/SFCS.1988.21948}.

\bibitemdeclare{article}{Savitch70}
\bibitem{Savitch70}
\bibinfo{author}{W.~J. \surnamestart Savitch\surnameend}
  (\bibinfo{year}{1970}): \emph{\bibinfo{title}{Relationships between
  nondeterministic and deterministic tape complexities}}.
\newblock {\sl \bibinfo{journal}{Journal of Computer and System Sciences}}
  \bibinfo{volume}{4}, pp. \bibinfo{pages}{177--192},
  \doi{10.1016/S0022-0000(70)80006-X}.

\bibitemdeclare{incollection}{thomas-salomaa}
\bibitem{thomas-salomaa}
\bibinfo{author}{W.~\surnamestart Thomas\surnameend} (\bibinfo{year}{1999}):
  \emph{\bibinfo{title}{Complementation of {B{\"u}}chi automata revisited}}.
\newblock In \bibinfo{editor}{J.~Karhum{{\"a}}ki \surnamestart
  et~al.\surnameend}, editor: {\sl \bibinfo{booktitle}{Jewels are Forever,
  Contributions on Theoretical Computer Science in Honor of Arto Salomaa}},
  \bibinfo{publisher}{Springer}, pp. \bibinfo{pages}{109--122},
  \doi{10.1007/978-3-642-60207-8\_10}.

\bibitemdeclare{inbook}{Vardi96}
\bibitem{Vardi96}
\bibinfo{author}{M.~Y. \surnamestart Vardi\surnameend} (\bibinfo{year}{1996}):
  \emph{\bibinfo{title}{An Automata-Theoretic Approach to Linear Temporal
  Logic}}, pp. \bibinfo{pages}{238--266}.
\newblock {\sl \bibinfo{series}{LNCS}} \bibinfo{volume}{1043},
  \bibinfo{publisher}{Springer}, \bibinfo{address}{New York, NY, USA},
  \doi{10.1007/3-540-60915-6\_6}.

\end{thebibliography}
